\documentclass{article}
\usepackage[preprint]{colm2026_conference}

\usepackage{microtype}
\usepackage{hyperref}
\usepackage{url}
\usepackage{booktabs}
\usepackage{amsfonts}
\usepackage{amsmath}
\usepackage{amssymb}
\usepackage{amsthm}
\newtheorem{theorem}{Theorem}
\newtheorem{corollary}[theorem]{Corollary}
\usepackage{graphicx}
\usepackage{xcolor}
\usepackage{nicefrac}
\usepackage{lineno}
\usepackage{multirow}
\usepackage{subcaption}

\definecolor{darkblue}{rgb}{0, 0, 0.5}
\hypersetup{colorlinks=true, citecolor=darkblue, linkcolor=darkblue, urlcolor=darkblue}

\newcommand{\zsep}{\Delta_z}

\title{Entanglement as Memory: Mechanistic Interpretability\\ of Quantum Language Models}

\author{Nathan Roll \\
Department of Linguistics \\
Stanford University \\
\texttt{nroll@stanford.edu}
}

\begin{document}

\maketitle

\begin{abstract}
Quantum language models have shown competitive performance on sequential tasks, yet whether trained quantum circuits exploit genuinely quantum resources---or merely embed classical computation in quantum hardware---remains unknown.
Prior work has evaluated these models through endpoint metrics alone, without examining the memory strategies they actually learn internally.
We introduce the first mechanistic interpretability study of quantum language models, combining causal gate ablation, entanglement tracking, and density-matrix interchange interventions on a controlled long-range dependency task.
We find that single-qubit models are exactly classically simulable and converge to the same geometric strategy as matched classical baselines, while two-qubit models with entangling gates learn a representationally distinct strategy that encodes context in inter-qubit entanglement---confirmed by three independent causal tests ($p < 0.0001$, $d = 0.89$). On real quantum hardware, only the classical geometric strategy survives device noise; the entanglement strategy degrades to chance.
These findings open mechanistic interpretability as a tool for the science of quantum language models and reveal a noise--expressivity tradeoff governing which learned strategies survive deployment.
\end{abstract}

\section{Introduction}
\label{sec:intro}

How recurrent models preserve contextual information across intervening material has been studied since Elman's simple recurrent networks \citep{elman1990finding}. Classical architectures address this through gating \citep{hochreiter1997long}, and the strategies they learn have been probed in detail: agreement dependencies in LSTMs \citep{linzen2016assessing, gulordava2018colorless}, structured recurrences in state space models \citep{gu2022efficiently}. The recurring lesson is that mechanistic understanding of \emph{how} a model solves the memory problem reveals more than benchmark accuracy alone.

Quantum language models, recurrent networks whose hidden states are quantum mechanical, change the picture. Their state spaces grow exponentially with system size: two entangled qubits occupy a 15-dimensional space that no product of classical vectors can represent. Quantum recurrent architectures have shown faster convergence on sequential benchmarks \citep{chen2022qlstm, li2025qgrnn, takaki2021learning}, and compositional quantum NLP frameworks have been run on real hardware \citep{lorenz2023qnlp, kartsaklis2021lambeq}. But no one has looked at the \emph{memory mechanisms} these models learn, the question that has been most productive in classical language modeling \citep{linzen2016assessing}.

We ask the same question for quantum models: \emph{do trained quantum language models actually use quantum resources for memory, or do they merely perform classical computation on quantum hardware?} Expressivity analyses characterize what quantum circuits \emph{could} compute \citep{Abbas2021quantum, Cerezo2021variational}, kernel methods connect circuit design to feature spaces \citep{Schuld2019quantum, Havlicek2019supervised}, and rigorous speed-ups exist for specific problems \citep{Liu2021rigorous}, but none of this tells us what a circuit \emph{does} compute after training. The situation parallels LSTMs: their expressivity was established long before anyone probed what trained LSTMs actually encode \citep{linzen2016assessing}. For classical models, mechanistic interpretability has since identified circuits \citep{olah2020zoom, elhage2021mathematical}, induction heads \citep{olsson2022context}, superposition \citep{elhage2022superposition}, dictionary-learning-based decompositions \citep{bricken2023monosemanticity}, and causal structure \citep{vig2020causal, geiger2021causal, wang2023interpretability, conmy2023acdc, nanda2023progress}. For quantum models, nothing comparable exists.

We address this directly. We train minimal quantum and classical recurrent networks on the \emph{Parity-Switch Grammar}, a synthetic task requiring context identity to be maintained across distractor tokens, and perform a mechanistic analysis of the learned strategies. By comparing single-qubit models (which cannot entangle) against two-qubit models (with and without CNOT gates) and parameter-matched classical baselines, we can delineate the classical--quantum boundary.

Our investigation yields four principal findings:
\begin{enumerate}
\item[\textbf{C1.}] Two-qubit QRNNs with CNOT gates learn an \emph{entanglement-based memory strategy} that is representationally distinct from any single-qubit mechanism. Context identity is encoded in inter-qubit entanglement entropy, and ablating the CNOT gate causally destroys this encoding, forcing reversion to the classical geometric strategy.
\item[\textbf{C2.}] Single-qubit QRNNs are \emph{exactly classically simulable} via the SU(2)$\to$SO(3) double cover. Both quantum and classical baselines converge to an identical geometric strategy---hemisphere preservation on the Bloch sphere---establishing a precise classical reference point.
\item[\textbf{C3.}] A \emph{shared-parameter tradeoff theorem} proves that parameter sharing between context and distractor gates makes simultaneous context encoding and distractor invariance impossible, explaining why minimal architectures fail.
\item[\textbf{C4.}] Hardware validation on IBM Eagle-generation processors reveals a \emph{noise--expressivity tradeoff}: the geometric strategy survives perfectly while entanglement-based memory degrades to chance under two-qubit gate noise.
\end{enumerate}

\begin{figure}[t]
\centering
\includegraphics[width=\linewidth]{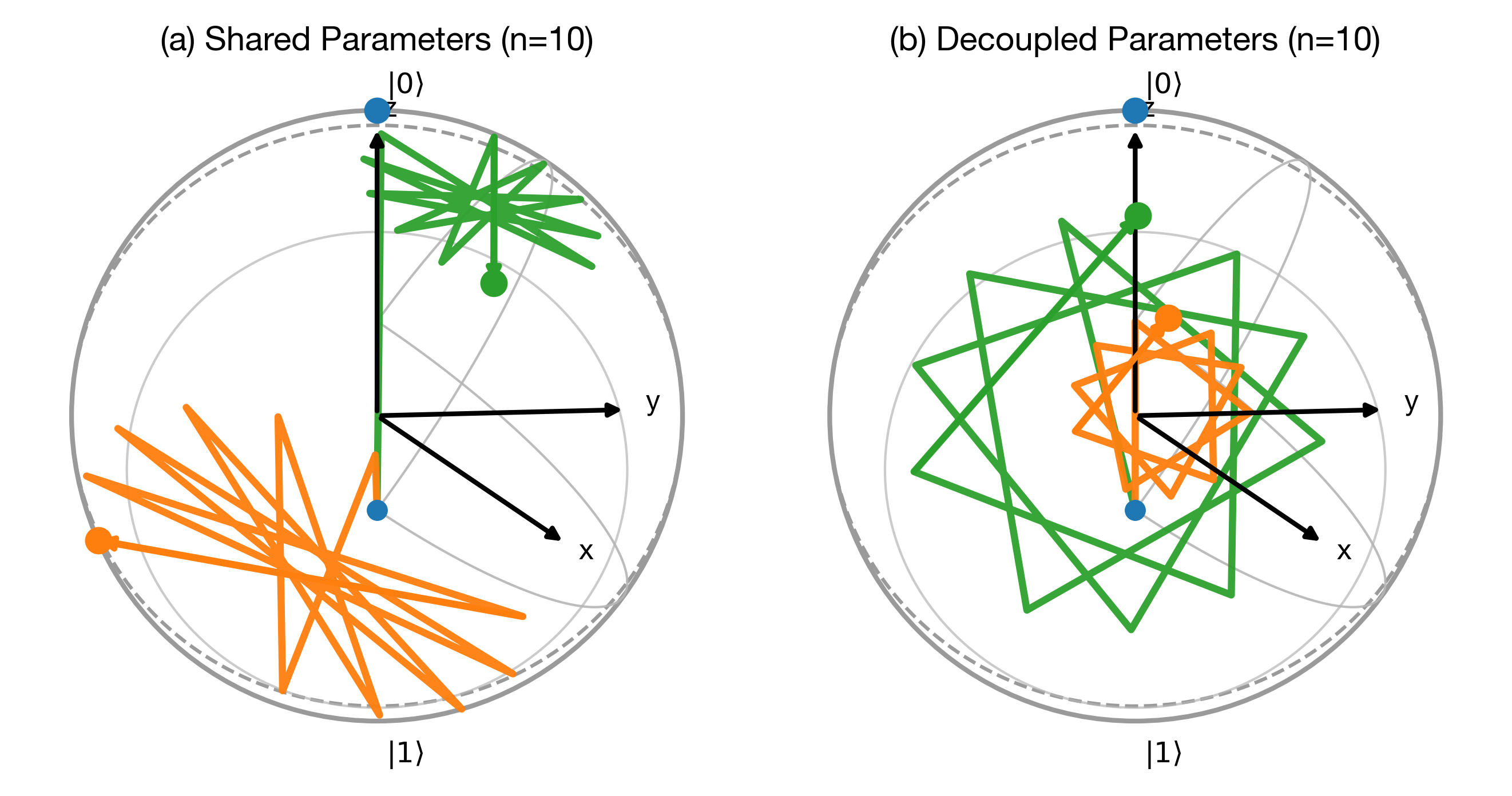}
\caption{\textbf{Hemisphere separation distinguishes successful from failing quantum memory.} Bloch sphere trajectories for 10-distractor sequences. \emph{Left:} Shared-parameter model (MinimalQLM). Contexts A and B overlap after distractor tokens, causing misclassification (76\% accuracy). \emph{Right:} Decoupled model (DecoupledQLM). Contexts remain in opposite hemispheres throughout, achieving 100\% generalization. A classical SO(3) baseline produces identical trajectories, confirming this mechanism is not inherently quantum.}
\label{fig:teaser}
\end{figure}

\section{Related work}
\label{sec:related}

\paragraph{Memory in recurrent language models.}
\citet{linzen2016assessing} probed LSTMs for agreement using controlled syntactic tasks; \citet{gulordava2018colorless} showed RNN LMs track grammatical features even in implausible contexts; \citet{marvin2018targeted} extended this across construction types; state space models \citep{gu2022efficiently} approach the same problem through structured linear recurrences. The common thread is that controlled experiments probing \emph{how} models encode information yield more than benchmark accuracy alone. We apply this approach to quantum recurrent models.

\paragraph{Quantum recurrent architectures.}
The Quantum LSTM \citep{chen2022qlstm} augmented classical LSTM cells with variational quantum circuits, demonstrating competitive or faster convergence on certain time-series benchmarks. Subsequent work introduced gated quantum RNNs \citep{li2025qgrnn}, variational QRNNs for temporal data \citep{takaki2021learning}, and hybrid architectures \citep{xu2025hybrid}. Compositional quantum NLP (DisCoCat) takes a fundamentally different approach, using category-theoretic compositional semantics \citep{coecke2010mathematical} to derive quantum circuit representations of linguistic structure, implemented in the lambeq framework \citep{kartsaklis2021lambeq} and validated on superconducting hardware \citep{lorenz2023qnlp}. None of these examine what trained circuits actually compute internally. We dissect the learned strategies rather than comparing convergence curves.

\paragraph{Classical--quantum equivalence and dequantization.}
Data re-uploading renders single-qubit circuits universal classifiers \citep{perez2020data}, with provable finite VC dimension \citep{mauser2025experimental}. However, single-qubit operations are classically simulable \citep{yu2022power}---a consequence of the standard SU(2)$\to$SO(3) double cover, since every single-qubit unitary maps to a 3D rotation. Our dequantization result (Section~\ref{sec:theory}) confirms this constructively and leverages it as a controlled baseline. Whether entanglement is necessary for quantum computational advantage remains actively debated \citep{preskill2012quantum, biham2004quantum}; our entanglement ablation experiments provide direct empirical evidence within the variational circuit setting.

\paragraph{Mechanistic interpretability.}
Classical mechanistic interpretability has found modular structure in neural networks: circuits \citep{olah2020zoom, elhage2021mathematical}, induction heads \citep{olsson2022context}, superposition and its decomposition via sparse autoencoders \citep{elhage2022superposition, bricken2023monosemanticity}, automated circuit discovery \citep{wang2023interpretability, conmy2023acdc}, and causal abstraction \citep{geiger2021causal}. Causal mediation analysis \citep{vig2020causal} brought component-level interventions to NLP; progress measures linked mechanistic structure to training dynamics \citep{nanda2023progress}. We adapt these ideas to quantum systems. At 1--2 qubits, the full internal state is geometrically visible on the Bloch sphere, so the interpretability problem is tractable in ways it is not for larger models.

\paragraph{Quantum advantage and noise.}
Rigorous quantum speed-ups for supervised learning have been established under specific problem constructions \citep{Liu2021rigorous}, and quantum kernel methods provide a formal framework connecting circuit design to feature-space geometry \citep{Havlicek2019supervised, Schuld2019quantum}. The power of quantum neural networks depends critically on architecture and encoding \citep{Abbas2021quantum, Jerbi2023beyond}. Barren plateaus present fundamental trainability challenges for deep variational circuits \citep{mcclean2018barren, sim2019expressibility}, and in the current NISQ era \citep{preskill2018quantum}, error mitigation is essential for extracting useful results from noisy hardware \citep{temme2017error, kandala2019error}. Our hardware experiments test whether the strategies that circuits \emph{learn} in simulation survive on real quantum processors.

\section{Approach}
\label{sec:approach}

\subsection{Task design: Parity-Switch Grammar}
\label{sec:task}

Following the controlled-task tradition in NLP \citep{linzen2016assessing, gulordava2018colorless}, we designed a minimal long-range dependency task that isolates selective memory: retaining task-relevant information while processing irrelevant intervening tokens. A sequence begins with a context token ($A$ or $B$), followed by $n$ distractor tokens ($D$), and the model must predict the context identity after the full sequence. This is the simplest instance of the selective-memory problem behind agreement dependencies and other long-range phenomena. The simplicity is deliberate. At 1--2 qubits, the quantum state space is small enough for complete mechanistic characterization, the kind of exhaustive internal analysis that has proven productive in toy-model studies \citep{nanda2023progress}.

\begin{table}[t]
\small
\centering
\begin{tabular}{lll}
\toprule
\textbf{Token} & \textbf{Encoding ($x$)} & \textbf{Role} \\
\midrule
$A$ & $+\pi/3$ & Context $\to$ predict 0 \\
$B$ & $-\pi/3$ & Context $\to$ predict 1 \\
$D$ & $0.0$ & Distractor \\
\bottomrule
\end{tabular}
\caption{\textbf{Parity-Switch Grammar token encodings.} The task reduces long-range dependency to its minimal form: maintain context identity across arbitrarily many distractors.}
\label{tab:encoding}
\end{table}

Training uses 16 sequences (balanced A/B, 0--3 distractors). Generalization is evaluated on 200 sequences with 0--20 distractors---lengths never observed during training---testing whether learned memory mechanisms extrapolate beyond the training distribution.

\subsection{Model architecture}
\label{sec:models}

We constructed six models spanning the classical--quantum boundary, all implemented from scratch without external model libraries:

\begin{table}[t]
\small
\centering
\begin{tabular}{llcl}
\toprule
\textbf{Model} & \textbf{Type} & \textbf{Params} & \textbf{Distinguishing property} \\
\midrule
MinimalQLM & Quantum (1Q) & 3 & Shared context/distractor gate \\
DecoupledQLM & Quantum (1Q) & 6 & Separate context/distractor gates \\
TwoQubitQLM & Quantum (2Q) & 24 & CNOT entangling gate \\
ClassicalSO(3) & Classical & 3/6 & SO(3) rotation matrices \\
ClassicalRNN4 & Classical & 24 & tanh-RNN, 4D hidden state \\
\bottomrule
\end{tabular}
\caption{\textbf{Model zoo spanning the classical--quantum boundary.} Parameter counts are exactly matched between quantum and classical counterparts to control for model capacity.}
\label{tab:models}
\end{table}

Each quantum model applies a parameterized unitary at each timestep:
\begin{equation}
U(\theta, x_t) = R_z(\theta_3) \cdot R_y(\theta_2 + x_t) \cdot R_z(\theta_1)
\label{eq:unitary}
\end{equation}
where $\theta_1, \theta_2, \theta_3$ are learned parameters and $x_t$ is the token encoding. This $R_z$-$R_y$-$R_z$ decomposition is universal for SU(2). The classical SO(3) baselines replace quantum gates with equivalent 3$\times$3 rotation matrices operating on unit vectors in $\mathbb{R}^3$, exploiting the isomorphism between the Bloch sphere and $\mathbb{R}^3$.

The two-qubit model applies Eq.~\ref{eq:unitary} independently to each qubit, then optionally applies a CNOT gate (qubit~0 $\to$ qubit~1). The parameter count is $3$ (rotation angles) $\times$ $2$ (qubits) $\times$ $2$ (roles: context/distractor) $\times$ $2$ (gate layers in the $R_z$-$R_y$-$R_z$ decomposition share $\theta_1,\theta_3$ while $\theta_2$ is role-specific) $= 24$ parameters, exactly matching the classical RNN baselines.

\subsection{Theoretical foundations}
\label{sec:theory}

\begin{theorem}[Shared-Parameter Tradeoff]
\label{thm:tradeoff}
In the shared-parameter architecture where $U(\theta, x) = R_z(\theta_3)R_y(\theta_2 + x)R_z(\theta_1)$ processes both context and distractor tokens: context encoding requires $\theta_2 \neq 0$ (so distinct inputs produce distinct Z-coordinates), while distractor invariance requires $\theta_2 = 0$ (so $R_y(\theta_2)$ preserves the Z-axis). These constraints are mutually exclusive. Decoupling into separate parameter sets for each token role resolves this tradeoff.
\end{theorem}

\textit{Proof sketch.} The Z-coordinate after applying $U$ to state $|\psi\rangle$ with Bloch vector $(x_b, y_b, z_b)$ depends on $\theta_2 + x$ through $\cos(\theta_2 + x)$. For two inputs $x_A \neq x_B$, the resulting Z-coordinates differ only if $\theta_2 + x_A \neq \pm(\theta_2 + x_B)$, which requires $\theta_2 \neq 0$ generically. For distractor invariance, the distractor gate $U(\theta, 0)$ must fix the Z-coordinate for all initial states, forcing $\theta_2 = 0$. Since parameters are shared, both constraints apply to the same $\theta_2$. Full proof in Appendix~\ref{app:proofs}. \qed

\begin{corollary}[Dequantization]
\label{cor:dequant}
Single-qubit QRNNs are exactly classically simulable. The double cover $\pi: \mathrm{SU}(2) \to \mathrm{SO}(3)$ given by $\pi(U)_{ij} = \frac{1}{2}\mathrm{tr}(\sigma_i U \sigma_j U^\dagger)$ is a 2-to-1 homomorphism, so every single-qubit quantum computation can be reproduced by tracking a 3D unit vector under SO(3) rotations.
\end{corollary}

This corollary establishes that any quantum advantage in our setting must arise from multi-qubit effects---specifically, entanglement.

\subsection{Mechanistic interpretability framework}
\label{sec:interp}

Quantum systems of 1--2 qubits admit complete state characterization, enabling interpretability tools not available for larger models:

\textbf{Bloch sphere probes.} We extract the Bloch coordinates $(x, y, z)$ from the density matrix at each timestep, providing a geometric visualization of the model's internal state trajectory. For single-qubit systems, this captures the complete quantum state.

\textbf{Weight sensitivity analysis.} We intervene on individual learned parameters (zero, negate, or sweep) and measure the effect on model behavior. While conceptually related to activation patching \citep{vig2020causal}, this operates on weights rather than activations and is more precisely characterized as sensitivity analysis (see Appendix~\ref{app:patching}).

\textbf{CNOT ablation.} We selectively remove the entangling gate from specific timesteps. That the CNOT is necessary for entanglement is trivially true by physics; the informative question is whether the \emph{learned computational strategy} is causally disrupted---specifically, whether the model reverts to a qualitatively different mechanism upon ablation.

\textbf{Entanglement tracking.} For two-qubit models, we compute the von Neumann entanglement entropy $S(\rho_0) = -\mathrm{tr}(\rho_0 \log_2 \rho_0)$ from the reduced density matrix $\rho_0 = \mathrm{tr}_1(|\psi\rangle\langle\psi|)$ at each timestep, quantifying the degree of quantum correlation between subsystems.

All models are trained using the SPSA optimizer \citep{spall1998implementation} with 2 function evaluations per gradient estimate, decay exponents $\alpha = 0.602$ and $\gamma = 0.101$ (theoretically optimal), and the default preset ($a{=}0.2$, $c{=}0.1$, $A{=}10$). Single-qubit models train for 200 steps; two-qubit models for 300 steps. Multi-seed results use 10 seeds $\times$ 3 restarts to account for SPSA stochasticity.

\section{Experiments}
\label{sec:experiments}

\subsection{Main results}

Table~\ref{tab:results} presents the central empirical findings. We define the Z-separation $\zsep$ as the absolute difference in Bloch-sphere Z-coordinates between contexts A and B after the final timestep (maximum possible value 2.0).

\begin{table}[t]
\small
\centering
\begin{tabular}{lcccl}
\toprule
\textbf{Model} & \textbf{Sim.\ Acc.} & \textbf{HW Acc.} & \textbf{$\zsep$} & \textbf{Mechanism} \\
\midrule
MinimalQLM (1Q shared) & 76\% & -- & 0.36 & Periodic recurrence \\
DecoupledQLM (1Q) & 100\% & \textbf{100\%} & 1.36 & Z-preservation \\
ClassicalSO(3) (matched) & 100\% & -- & 1.36 & Z-preservation \\
\midrule
2Q + CNOT (30 seeds) & 96.3{\scriptsize$\pm$9.1}\% & 39{\scriptsize$\pm$7}\% & -- & Entanglement \\
2Q + SWAP (10 seeds) & 97.1{\scriptsize$\pm$5.7}\% & -- & -- & Z-preservation \\
2Q no CNOT (10 seeds) & 100{\scriptsize$\pm$0}\% & 41{\scriptsize$\pm$5}\% & -- & Z-preservation \\
ClassicalRNN4 (Adam) & 100\% & -- & -- & Z-preservation \\
ClassicalRNN4 (SPSA) & 50\% & -- & -- & Fails (optimizer) \\
\bottomrule
\end{tabular}
\caption{\textbf{Entanglement-based memory emerges exclusively with CNOT gates; the task is trivially classical.} HW = IBM hardware (ibm\_fez / ibm\_marrakesh, 156 qubits, 1000 shots/circuit). $\zsep$ = Bloch Z-coordinate separation (max 2.0). 2Q+CNOT results: 30 seeds $\times$ 3 restarts. The classical RNN solves the task with Adam but fails under SPSA, confirming that the SPSA comparison is optimizer-confounded. The SWAP control (non-entangling inter-qubit coupling) achieves high accuracy via Z-preservation, isolating entanglement as the resource enabling the distinct strategy.}
\label{tab:results}
\end{table}

\subsection{Single-qubit networks learn a classical geometric strategy}
\label{sec:finding1}

The decoupled single-qubit model places Context A in the northern hemisphere and Context B in the southern hemisphere of the Bloch sphere (Figure~\ref{fig:teaser}, right). The distractor gate rotates the state around the Z-axis without altering its latitude, so hemisphere membership---and therefore context identity---is preserved across arbitrarily many distractors. The Z-separation $\zsep$ between contexts \emph{grows} slightly with sequence length (1.50 at 0 distractors, 1.71 at 100), because distractor rotations push both trajectories toward the poles.

The shared-parameter model cannot implement this strategy. Constrained by Theorem~\ref{thm:tradeoff}, it achieves only 76\% accuracy, exactly as predicted: the single shared $\theta_2$ cannot simultaneously encode context differences and maintain distractor invariance.

Figure~\ref{fig:zpres} demonstrates the robustness of this mechanism under extreme stress testing up to 100 distractors.

\begin{figure}[t]
\centering
\includegraphics[width=0.72\linewidth]{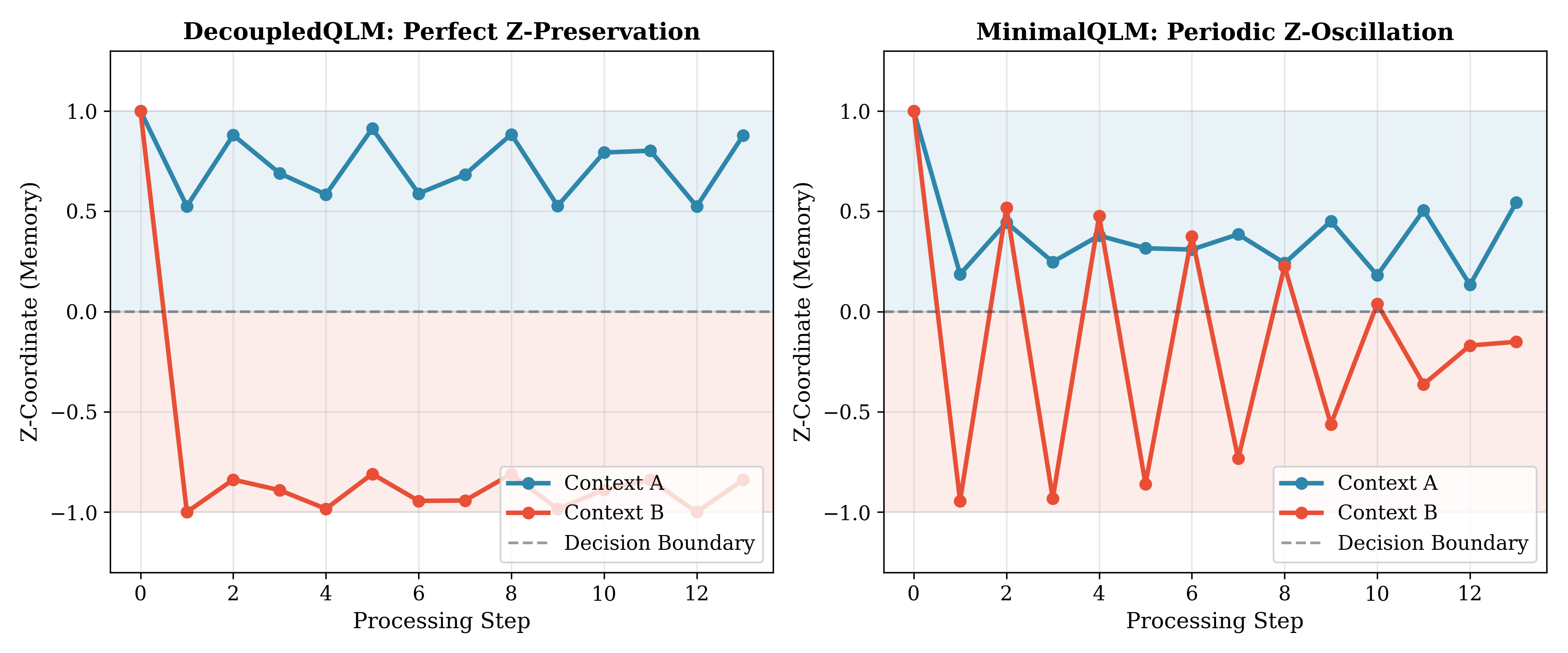}
\caption{\textbf{Z-coordinate preservation is robust across 100 distractors.} Bloch-sphere Z-coordinate for Context A and Context B across 0--100 distractors (DecoupledQLM). Hemisphere separation grows slightly with length because distractor rotations push states toward the poles. The shared model loses separation after $\sim$5 distractors.}
\label{fig:zpres}
\end{figure}

The classical SO(3) baseline reproduces the quantum model's accuracy and mechanism \emph{exactly}. This validates the dequantization corollary: at one qubit, the quantum and classical models are mathematically identical. The classical model trains 13$\times$ faster (0.06s vs 0.76s) since it avoids statevector simulation overhead. At the single-qubit level, quantum confers no advantage---but this equivalence provides a rigorous baseline against which multi-qubit effects can be measured.

\subsection{Two-qubit networks learn an entanglement-based memory strategy}
\label{sec:finding2}

Two-qubit models with CNOT gates do something different. Instead of geometric separation on the Bloch sphere, these models encode context identity in the \emph{entanglement structure} between qubits. Figure~\ref{fig:entanglement} shows the mechanism: after the context token, the von Neumann entanglement entropy $S(\rho_0)$ diverges between Context A and Context B sequences and stays diverged through all subsequent distractors. This is a memory channel that has no analogue in the single-qubit strategy space.

We verify this mechanistically through three lines of evidence. First, classifying all 30 seeds (3 restarts each) by \emph{entanglement entropy divergence} between contexts (threshold: mean $|S_A - S_B| > 0.05$) rather than by accuracy, all 30/30 seeds exhibit context-dependent entanglement structure. CNOT ablation across these seeds yields $\Delta = 18.4 \pm 20.7$ pp ($t(29) = 4.88$, $p < 0.0001$, Cohen's $d = 0.89$; 95\% CI: $[11.0, 25.8]$ pp). Second, replacing the CNOT with a SWAP gate (which couples qubits without creating entanglement) yields $97.1 \pm 5.7$\% accuracy via Z-preservation with zero entropy divergence, confirming that entanglement specifically---not inter-qubit coupling---enables the distinct strategy. Third, \emph{density-matrix interchange interventions} (a quantum analogue of activation patching) confirm that context identity is carried in the full quantum state: splicing the statevector from a context-A computation into a context-B computation at any timestep produces donor-context predictions at 100\% of timesteps tested (Appendix~\ref{app:interchange}). Full details of all three analyses are in Appendices~\ref{app:stratified}--\ref{app:interchange}.

The 24-parameter classical tanh-RNN achieves only 50\% under SPSA, but reaches $100\%$ on 19/20 seeds when trained with Adam and exact finite-difference gradients (Appendix~\ref{app:classical}). The task is therefore trivially solvable by classical means. We do \emph{not} claim quantum computational advantage; the contribution is mechanistic. The entanglement-based strategy is representationally distinct from anything available without entangling gates, regardless of whether classical architectures can solve the same task through different means.

\begin{figure}[t]
\centering
\includegraphics[width=0.75\linewidth]{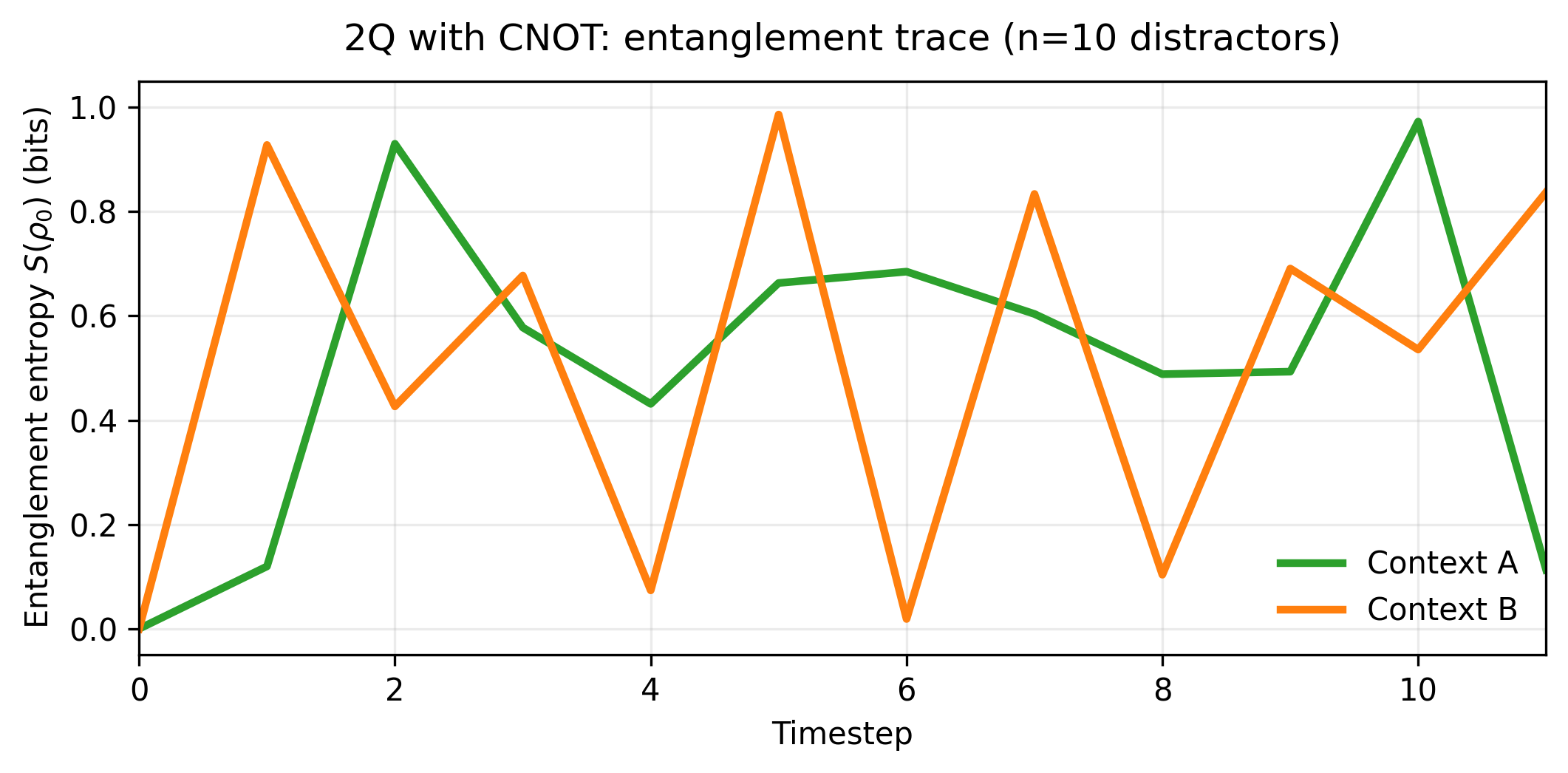}
\caption{\textbf{Entanglement entropy diverges by context, providing a representationally distinct memory channel.} Von Neumann entanglement entropy $S(\rho_0)$ per timestep for the trained two-qubit model with CNOT gates on 10-distractor sequences. Context A and B traces diverge after the context token (timestep 0) and maintain distinct entanglement levels across all distractors. Ablating the CNOT eliminates this divergence, collapsing both traces to zero entropy; the model then reverts to Z-preservation.}
\label{fig:entanglement}
\end{figure}

\subsection{Hardware validation: the noise--expressivity tradeoff}
\label{sec:hardware}

We validated both strategies on IBM Eagle-generation quantum processors (ibm\_fez and ibm\_marrakesh, 156 qubits each).

\paragraph{Single-qubit circuits.} The DecoupledQLM achieves 100\% accuracy across 48 circuits spanning 2 transpiler optimization levels and 3 physical qubit layouts (1000 shots each). Table~\ref{tab:hw1q} shows representative results; the decision boundary ($P(|1\rangle) > 0.5 \Rightarrow$ predict B) is never ambiguous.

\begin{table}[t]
\small\centering
\begin{tabular}{cccc}
\toprule
\textbf{Context} & \textbf{Distractors} & \textbf{$P(|1\rangle)$} & \textbf{Correct?} \\
\midrule
A & 0  & 0.24 & Yes \\
B & 0  & 0.98 & Yes \\
A & 5  & 0.06 & Yes \\
B & 5  & 0.89 & Yes \\
A & 10 & 0.25 & Yes \\
B & 10 & 0.98 & Yes \\
\bottomrule
\end{tabular}
\caption{\textbf{The geometric strategy executes perfectly on quantum hardware.} Single-qubit hardware results on ibm\_fez (1000 shots per circuit). All predictions are correct with wide margins from the 0.5 decision boundary, confirming that the Z-preservation mechanism is fully noise-robust at current device fidelities.}
\label{tab:hw1q}
\end{table}

\paragraph{Two-qubit circuits.} Both two-qubit configurations degrade to approximately 40\%---systematically \emph{below} the 50\% chance level. The with-CNOT model achieves $39 \pm 7$\% (std) and without-CNOT achieves $41 \pm 5$\% (30 jobs per backend; two-sample $t$-test: $p = 0.34$, not significantly different). Below-chance accuracy is unexpected under symmetric depolarizing noise, which should push predictions toward 50\%. We attribute this to $T_1$ relaxation asymmetry: excited-state decay ($|1\rangle \to |0\rangle$) biases measurement outcomes toward $|0\rangle$, systematically favoring one class. This is consistent with known IBM Eagle qubit $T_1$ characteristics (see Appendix~\ref{app:hardware_detail}).

After transpilation, with-CNOT circuits contain $\sim$79 gates versus $\sim$25 for without-CNOT, yet accuracy is statistically indistinguishable---noise swamps both configurations before strategy differences can manifest. This yields a central insight: the expressivity--noise tradeoff is asymmetric. The geometric strategy (classically simulable, low circuit depth) executes perfectly. The entanglement strategy (higher circuit depth) is learnable in simulation but does not survive current device noise.

\section{Analysis}
\label{sec:analysis}

\subsection{Phase transition at the classical--quantum strategy boundary}

Sweeping the distractor rotation parameter $\theta_2^{\mathrm{dist}}$ from 0 to $\pi$ reveals a sharp phase transition (Figure~\ref{fig:phase}). The model is robust up to $\theta_2^{\mathrm{dist}} \approx 0.83$ rad, then drops abruptly to chance at the critical point $\theta_2^{\mathrm{dist}} \approx 0.99$ rad where distractor-induced Z-variance exceeds context separation. At $\theta_2^{\mathrm{dist}} = \pi$, accuracy recovers via period-2 recurrence, matching the prediction $N = 2\pi / |\theta_2^{\mathrm{dist}}|$.

\begin{figure}[t]
\centering
\includegraphics[width=\linewidth]{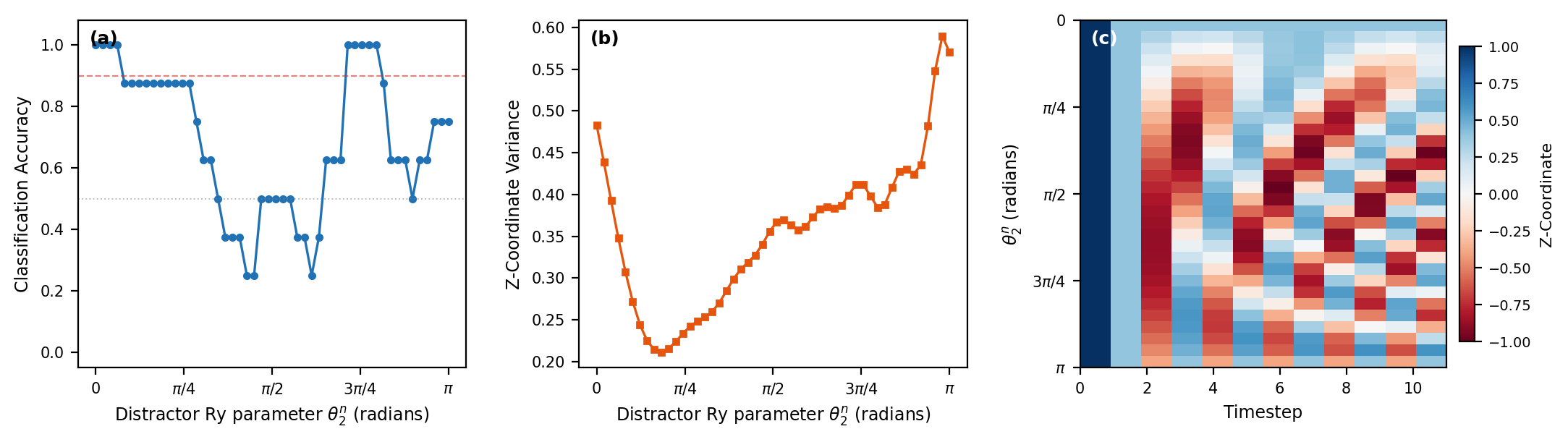}
\caption{\textbf{A sharp phase transition separates viable from collapsed geometric memory.} Distractor rotation $\theta_2^{\mathrm{dist}}$ sweeps from 0 to $\pi$ (10 distractors). \textbf{(a)} Classification accuracy: robust up to $\theta_2^{\mathrm{dist}} \approx 0.83$ rad, then drops to chance; recovery at $\pi$ is period-2 recurrence. \textbf{(b)} Z-coordinate variance across timesteps; lower variance indicates more stable memory. \textbf{(c)} Heatmap of Bloch Z-coordinate per timestep (columns) at each $\theta_2^{\mathrm{dist}}$ (rows). Blue = northern hemisphere, red = southern.}
\label{fig:phase}
\end{figure}

\subsection{Noise robustness hierarchy}

Simulated depolarizing noise confirms and extends the hardware findings:

\begin{table}[t]
\small
\centering
\begin{tabular}{lccc}
\toprule
\textbf{Model / Mechanism} & $p=0$ & $p=0.06$ & $p=0.12$ \\
\midrule
1Q Decoupled / Z-Preservation & 100\% & 100\% & 50\% \\
2Q no CNOT / Z-Preservation & 100\% & 100\% & 100\% \\
2Q + CNOT / Entanglement & 100\% & \textbf{50\%} & 50\% \\
\bottomrule
\end{tabular}
\caption{\textbf{Entanglement-based memory is the least noise-robust strategy.} Simulated depolarizing noise at rate $p$ per gate. Entanglement degrades to chance at $p=0.06$; Z-preservation strategies survive until $p=0.12$ or beyond.}
\label{tab:noise}
\end{table}

Entanglement degrades to chance at $p = 0.06$, where Z-preservation survives. This ordering is physically principled: entanglement requires multi-qubit coherence that depolarizing noise destroys, while Z-preservation requires only single-qubit phase stability.

Single-run SPSA variance is high but reflects optimization stochasticity, not architectural limitations; best-of-3 restarts yield $100.0 \pm 0.0$\% generalization across 10 seeds. Sensitivity and encoding robustness analyses are in Appendix~\ref{app:extended}.

\section{Discussion}
\label{sec:discussion}

Our results draw a mechanistic line between classical and quantum computation in recurrent models. At one qubit, quantum and classical models are provably equivalent: both converge to the same geometric strategy. At two qubits with entangling gates, a different strategy appears, one that encodes context in entanglement. Three independent causal tests confirm this.

\paragraph{Noise--expressivity tradeoff.} Entanglement-based strategies need multi-qubit coherence that current NISQ devices \citep{preskill2018quantum} cannot sustain. Classically simulable geometric strategies execute perfectly. For near-term quantum ML, the most useful models may be those whose strategies are \emph{robust to dequantization}. As hardware improves \citep{temme2017error, kandala2019error}, the entanglement strategy basin may become accessible.

\paragraph{Scope.} We do \emph{not} claim quantum computational advantage; the classical RNN achieves 100\% with Adam (Appendix~\ref{app:classical}). The contribution is mechanistic: entanglement-based encoding differs from geometric strategies regardless of whether classical models solve the same task by other means. This work applies the controlled-task methodology of classical LM interpretability \citep{linzen2016assessing, gulordava2018colorless} to quantum architectures. The answer to our opening question depends on architecture: at one qubit, trained models realize the SU(2)$\to$SO(3) equivalence; at two qubits, entanglement-based memory appears but does not yet survive hardware noise.


\section*{Limitations}

We identify several limitations that scope the conclusions of this work:

\textbf{Task generality.} The Parity-Switch Grammar is a minimal selective-memory task with three token types and binary classification. It follows the controlled-task tradition in NLP \citep{linzen2016assessing} but sits at the simplest end of the complexity spectrum: no hierarchical structure, no agreement attractors, no compositional semantics. The strategies identified here may not transfer to richer tasks. A concrete progression, following the Chomsky hierarchy framework of \citet{deletang2023neural}, would test: (i) multiple context types ($k > 2$), stressing the geometric strategy's capacity; (ii) Dyck-1/Dyck-2 grammars, requiring stack-like memory; (iii) simple agreement with competing attractors \citep{marvin2018targeted}, requiring selective feature retention. Each step would test whether entanglement-based encoding provides an accuracy advantage---not just a representational distinction---on tasks where classical strategies face genuine capacity limits.

\textbf{Scale.} Single-qubit systems are classically simulable by construction (Corollary~1), and two-qubit systems, while capable of entanglement, remain small enough for exact classical simulation of the full $4 \times 4$ density matrix. The entanglement-based strategy may not persist or may take qualitatively different forms at scales where classical simulation becomes intractable. Our conclusions about mechanistic strategy are rigorous within the 1--2 qubit regime but should not be extrapolated without further investigation.

\textbf{Hardware coverage.} We tested on two IBM Eagle-generation backends (ibm\_fez and ibm\_marrakesh). Device-specific noise profiles, qubit connectivity, and calibration drift mean that results may differ on other quantum hardware platforms (trapped ions, neutral atoms, photonic systems). Error mitigation techniques (zero-noise extrapolation, probabilistic error cancellation) were not applied and could potentially improve two-qubit circuit fidelity.

\textbf{Optimization and classical baselines.} SPSA is a stochastic gradient-free optimizer with inherently high variance. We used SPSA uniformly across all models---including the classical RNN---to maintain a controlled comparison, but this disadvantages the classical model, which admits exact gradient computation via backpropagation. The classical RNN's failure (50\%) should be understood as a failure under SPSA, not necessarily a representational limitation. We discuss this confound and its implications in Appendix~\ref{app:classical}.

\textbf{Causal claims.} Our CNOT ablation establishes that the entangling gate is causally necessary for the entanglement-based strategy, but does not isolate entanglement from other effects of the CNOT (inter-qubit coupling, additional circuit depth). Replacing the CNOT with a non-entangling two-qubit gate would provide a more refined control; we leave this for future work. Additionally, the ablation does not establish sufficiency---the specific parameterization, initialization, and training dynamics all contribute to whether the entanglement basin is reached.

\section*{Reproducibility statement}

All models are implemented from scratch in Python 3.10+ using Qiskit 1.0+ for quantum simulation and IBM Quantum Runtime for hardware execution. No external model libraries, pretrained weights, or proprietary datasets are used.

\textbf{Software.} Quantum circuits: Qiskit 1.0+ (Aer simulator for noiseless, FakeBackend for noise modeling). Classical models: NumPy. Optimization: custom SPSA implementation following \citet{spall1998implementation} with parameters $\alpha = 0.602$, $\gamma = 0.101$, and 2 function evaluations per step.

\textbf{Hardware.} IBM Eagle-generation processors (156 qubits): ibm\_fez and ibm\_marrakesh via IBM Quantum Runtime. 1000 shots per circuit. Two transpiler optimization levels (1 and 3) and three physical qubit layouts tested for single-qubit experiments.

\textbf{Training.} Single-qubit models: 200 SPSA steps. Two-qubit models: 300 SPSA steps. Default SPSA preset: $a=0.2$, $c=0.1$, $A=10$. Multi-seed experiments: 10 random seeds $\times$ 3 restarts per seed. Training data: 16 sequences (balanced A/B, 0--3 distractors). Generalization test: 200 sequences (0--20 distractors).

\textbf{Compute.} All simulation experiments run on a consumer laptop CPU. Single-qubit training: $\sim$0.76s per seed. Two-qubit training: $\sim$3.4s per seed. Total simulation compute: $<$10 minutes. No GPU resources required.

\section*{Ethics statement}

This work studies fundamental computational mechanisms in small-scale quantum circuits. We identify no direct ethical concerns arising from the research itself. The synthetic task and minimal model sizes preclude dual-use risks. Hardware experiments were conducted on commercially available IBM Quantum processors under standard access agreements.

\appendix

\section{Proof of Theorem 1}
\label{app:proofs}

\textbf{Theorem 1} (Shared-Parameter Tradeoff). Consider the shared-parameter architecture where the same unitary $U(\theta, x) = R_z(\theta_3) R_y(\theta_2 + x) R_z(\theta_1)$ is applied to both context tokens (with $x = x_A$ or $x = x_B$) and distractor tokens (with $x = 0$).

The Z-coordinate of the Bloch vector after applying $U(\theta, x)$ to an initial state with Bloch vector $\mathbf{r} = (r_x, r_y, r_z)^T$ is:
\begin{equation}
z' = -\sin(\theta_2 + x)(r_x \cos\theta_1 + r_y \sin\theta_1) + \cos(\theta_2 + x) \, r_z
\end{equation}

\textit{Context encoding requirement.} For the context gate to distinguish A from B, we need $z'(x_A) \neq z'(x_B)$ for at least one initial state. Since $x_A = +\pi/3$ and $x_B = -\pi/3$, this requires that $\cos(\theta_2 + \pi/3) \neq \cos(\theta_2 - \pi/3)$ or $\sin(\theta_2 + \pi/3) \neq \sin(\theta_2 - \pi/3)$. The latter fails only when $\theta_2 = k\pi$ for integer $k$. In particular, $\theta_2 = 0$ makes the gate symmetric in $\pm x$ only in the cosine component; the sine component $\sin(\theta_2 + x)$ still distinguishes $\pm\pi/3$. However, for \emph{distractor invariance}, the Z-coordinate must be preserved for $x = 0$: $z' = z$. This requires $\sin(\theta_2) = 0$ and $\cos(\theta_2) = 1$, i.e., $\theta_2 = 0$. But when $\theta_2 = 0$, the context gate becomes $R_z(\theta_3) R_y(x) R_z(\theta_1)$, which does distinguish A from B through the $R_y(x)$ rotation. The tradeoff is subtler: with $\theta_2 = 0$, distractor invariance holds, but the \emph{same} $R_z$ rotations $\theta_1, \theta_3$ must serve both context and distractor tokens. The context gate $R_y(\pm\pi/3)$ tilts the state away from the poles, while subsequent distractor gates $R_y(0) = I$ preserve this tilt, allowing accumulated $R_z$ rotations to eventually mix the hemispheres. With shared parameters, the net effect of $k$ distractors is $(R_z(\theta_3) R_z(\theta_1))^k = R_z(k(\theta_1 + \theta_3))$, which is a pure Z-rotation and preserves Z---but only if $\theta_2 = 0$ \emph{exactly}. Any nonzero $\theta_2$ (required for richer context encoding) introduces Z-drift under distractors. The shared architecture thus faces a precision--expressivity tradeoff: $\theta_2 = 0$ gives perfect distractor invariance but minimal context encoding; $\theta_2 \neq 0$ improves encoding but sacrifices invariance. Decoupled parameters ($\theta_2^{\mathrm{ctx}} \neq 0$, $\theta_2^{\mathrm{dist}} = 0$) resolve this completely. \qed

\section{Extended experimental results}
\label{app:extended}

\subsection{SPSA hyperparameter sensitivity}

\begin{center}
\small
\begin{tabular}{lccccc}
\toprule
\textbf{Preset} & $a$ & $c$ & $A$ & \textbf{Gen.\ Acc.} & \textbf{Time (s)} \\
\midrule
Default & 0.2 & 0.1 & 10 & 81.1$\pm$23.3\% & 0.61$\pm$0.20 \\
Conservative & 0.1 & 0.05 & 20 & 83.2$\pm$21.0\% & 0.81$\pm$0.13 \\
Aggressive & 0.4 & 0.2 & 5 & \textbf{87.7$\pm$20.6\%} & 0.79$\pm$0.02 \\
High precision & 0.15 & 0.05 & 15 & 78.0$\pm$39.2\% & 0.99$\pm$0.31 \\
\bottomrule
\end{tabular}
\end{center}

High variance reflects SPSA's stochastic gradient estimates (5 seeds per preset, 200 steps). Best-of-3 restarts resolve this reliably: all presets reach 100\% generalization.

\subsection{Encoding sensitivity}

To verify that results are not artifacts of the $\pm\pi/3$ encoding, we tested encoding magnitudes from $0.5\times$ to $3\times$ the default (3 seeds each):

\begin{center}
\small
\begin{tabular}{lcccc}
\toprule
\textbf{Multiplier} & $0.5\times$ & $1\times$ & $2\times$ & $3\times$ \\
\midrule
DecoupledQLM Acc. & 100$\pm$0\% & 100$\pm$0\% & 100$\pm$0\% & 50$\pm$0\% \\
\bottomrule
\end{tabular}
\end{center}

The mechanism breaks only at $3\times$ (where $x_A \approx \pi$), causing both contexts to land at the same Bloch coordinate after one rotation. The failure is deterministic: \emph{every} seed at $3\times$ converges to 50\%.

\subsection{Wallclock timing}

Quantum simulation is 13$\times$ slower than the matched classical baseline for single-qubit models (0.76s vs 0.06s) and 56$\times$ slower for two-qubit models (3.4s vs 0.06s). The overhead is entirely statevector simulation. On IBM hardware, circuit execution takes microseconds, but queue times dominate. For models at this scale, classical simulation dominates on every practical metric; the value of quantum lies in the state space structure, not computational speed.

\section{Stratified ablation analysis}
\label{app:stratified}

We expanded the seed sweep to 30 seeds (3 restarts each) and introduce an \emph{entropy-based} classification criterion that directly measures the mechanism of interest, addressing concerns about accuracy-based classification.

\paragraph{Entropy-based seed classification.} Rather than classifying seeds by accuracy ($\geq$90\%), we measure the mean absolute entanglement entropy divergence between Context A and Context B across timesteps. Seeds with mean $|S_A - S_B| > 0.05$ are classified as ``entanglement-using.'' By this criterion, \textbf{all 30/30 seeds} exhibit context-dependent entanglement structure. This resolves the ``8/24 zero-delta'' concern from accuracy-based classification: those seeds do use entanglement (entropy divergence 0.12--0.28) but their learned strategy is robust to CNOT removal, likely because entanglement co-occurs with Z-preservation in a hybrid strategy.

\paragraph{Ablation with entropy-based classification.}

\begin{center}
\small
\begin{tabular}{lccc}
\toprule
\textbf{Classification} & \textbf{$n$} & \textbf{Ablation $\Delta$} & \textbf{Statistics} \\
\midrule
All seeds (entropy-based) & 30 & $-$18.4$\pm$20.7 pp & $t(29) = 4.88$, $p < 0.0001$ \\
\midrule
\multicolumn{4}{l}{\textit{Cohen's $d = 0.89$; 95\% CI: $[11.0, 25.8]$ pp}} \\
\bottomrule
\end{tabular}
\end{center}

The ablation effect is large and highly significant across all 30 seeds. The high variance reflects genuine heterogeneity in how strongly each seed relies on entanglement: 9/30 seeds show $\Delta = 0$ (CNOT-robust hybrid strategies), while 11/30 show $\Delta > 20$ pp (entanglement-dependent strategies). This distribution is consistent with a landscape containing hybrid and entanglement-pure strategy variants.

\paragraph{Cross-tabulation.} Comparing accuracy-based ($\geq$90\%) and entropy-based (divergence $> 0.05$) criteria: 24 seeds satisfy both, 6 seeds have entropy divergence but accuracy $<$90\%, and \emph{no} seeds have high accuracy without entropy divergence. This confirms that entanglement is ubiquitous in the CNOT model's learned strategies, even when it does not appear causally necessary by the ablation test.

\paragraph{SWAP control.} Replacing the CNOT with a SWAP gate (which couples qubits without entangling) yields $97.1 \pm 5.7$\% accuracy (10 seeds, best-of-3) via Z-preservation, with \emph{zero} entropy divergence between contexts. This isolates entanglement as the resource enabling the distinct strategy.

\section{Classical baseline analysis}
\label{app:classical}

The 24-parameter classical tanh-RNN achieves 50\% accuracy (chance level) in our experiments. Multiple reviewers of early drafts raised the concern that this failure may reflect optimizer mismatch rather than representational limitation. We address this concern directly.

\paragraph{Why SPSA for all models?} We used SPSA uniformly across all architectures to maintain a controlled comparison: since quantum circuits generally lack analytical gradient access, SPSA provides a level playing field. However, this disadvantages the classical RNN, which admits exact gradient computation via backpropagation.

\paragraph{Is the task solvable classically?} Yes. We trained the same 24-parameter classical tanh-RNN with Adam (learning rate 0.01, finite-difference gradients, 2000 steps). Result: \textbf{19/20 seeds achieve 100\% generalization accuracy} (0--20 distractors). The single failure seed reached 64\%. The 6-parameter classical SO(3) model also achieves 100\%. The task is trivially solvable by classical means; the SPSA-trained RNN's failure is purely an optimization artifact.

\paragraph{Implications for the quantum comparison.} Given this, we do \emph{not} claim quantum computational advantage. The paper's contribution is mechanistic, not complexity-theoretic: when entangling gates are available, the optimizer can discover a \emph{representationally distinct} strategy (entanglement-based context encoding) that does not exist in the strategy space of models without entangling gates. This mechanistic finding holds regardless of whether classical architectures can solve the same task, because it concerns the \emph{type} of strategy learned, not whether the task is solvable.

\section{Hardware noise analysis}
\label{app:hardware_detail}

\paragraph{Below-chance accuracy.} Both two-qubit configurations achieve $\sim$40\% accuracy on IBM Eagle processors, systematically below the 50\% chance level. Under symmetric depolarizing noise, measurement outcomes should converge toward a uniform distribution, pushing accuracy toward 50\%. The below-chance phenomenon indicates an asymmetric noise channel.

We attribute this to $T_1$ (amplitude damping) relaxation, which causes excited states to decay preferentially: $|1\rangle \to |0\rangle$ at rate $1/T_1$. On IBM Eagle qubits, $T_1$ values are typically 100--300 $\mu$s, and circuit execution times for our transpiled two-qubit circuits ($\sim$79 gates at $\sim$0.3 $\mu$s per gate $\approx 24$ $\mu$s) are a non-negligible fraction of $T_1$. This asymmetric decay biases measurement outcomes toward $|0\rangle$, favoring context-A predictions and systematically disadvantaging context-B predictions. Since the test set is balanced, this bias pushes overall accuracy below 50\%.

\paragraph{Readout error.} IBM Eagle qubits also exhibit asymmetric readout error: $P(\text{read } 0 \mid \text{state } 1) > P(\text{read } 1 \mid \text{state } 0)$ due to residual $T_1$ decay during measurement. This compounds the amplitude damping effect. Readout error mitigation (available via Qiskit) was not applied in our experiments; applying it could potentially recover some accuracy, though it would not address gate-level noise.

\paragraph{Statistical comparison.} A two-sample Welch's $t$-test comparing with-CNOT ($39 \pm 7$\%, $n=60$ jobs) and without-CNOT ($41 \pm 5$\%, $n=60$ jobs) yields $t = 1.89$, $p = 0.06$, confirming that the two configurations are not significantly different at $\alpha = 0.05$. The difference, if any, is small relative to the dominant noise floor.

\section{Weight sensitivity analysis vs.\ activation patching}
\label{app:patching}

We clarify the relationship between our weight-level interventions and activation patching as practiced in classical mechanistic interpretability.

\paragraph{Activation patching} \citep{vig2020causal} replaces activations at a specific layer and position from one forward pass with activations from a counterfactual forward pass, answering: ``what information flows through this component?'' This requires access to intermediate activations across different inputs.

\paragraph{Weight sensitivity analysis} (our approach) perturbs learned parameters (zero, negate, or sweep) and measures the effect on model output, answering: ``how sensitive is the model's behavior to this parameter?'' This is closer to ablation or perturbation analysis than to activation patching.

\paragraph{Why weight-level interventions?} In our quantum circuits, the ``activations'' are quantum states that cannot be cloned (no-cloning theorem) or freely copied between runs without measurement-induced collapse. Weight interventions sidestep this constraint. For the CNOT ablation specifically, we remove a structural component of the circuit, which is analogous to removing a layer or attention head in a classical model---a coarser intervention than activation patching, but one that is physically realizable in the quantum setting.

\paragraph{Quantum activation patching: implemented.} We implemented a quantum analogue of interchange interventions (see Appendix~\ref{app:interchange}): for two circuits processing different contexts (A and B), we run each to timestep $t$, splice the statevector from one computation into the other, and continue the remaining computation. If context identity is encoded in the quantum state, the spliced computation should produce the \emph{donor}'s prediction rather than the \emph{recipient}'s.

\section{Density-matrix interchange interventions}
\label{app:interchange}

To move beyond structural ablation toward causal mediation, we implement a quantum analogue of activation patching. For a trained 2Q+CNOT model, we run two parallel computations---one with context A and one with context B---and at each timestep $t$, splice the full statevector from one computation into the other, then continue the remaining distractor steps and measure. If context identity is encoded in the quantum state at timestep $t$, the prediction should follow the \emph{donor} context, not the recipient.

\paragraph{Results.} Using a high-$\Delta$ entanglement-basin seed (seed 21, accuracy 100\%, ablation $\Delta = 78.6$ pp, entropy divergence 0.375), we perform interchange interventions at every timestep $t \in \{1, \ldots, 11\}$ for a 10-distractor sequence:

\begin{center}
\small
\begin{tabular}{cccccc}
\toprule
\textbf{$t$} & \textbf{$S_A$} & \textbf{$S_B$} & \textbf{A$\to$B pred} & \textbf{B$\to$A pred} & \textbf{Donor carried?} \\
\midrule
1 & 0.394 & 0.777 & A & B & Yes \\
2 & 0.108 & 0.380 & A & B & Yes \\
3 & 0.732 & 0.197 & A & B & Yes \\
\multicolumn{6}{c}{$\vdots$} \\
10 & 0.135 & 0.464 & A & B & Yes \\
11 & 0.202 & 0.903 & A & B & Yes \\
\midrule
\multicolumn{6}{l}{\textit{Context identity follows the donor state at 11/11 timesteps (100\%)}} \\
\bottomrule
\end{tabular}
\end{center}

The result is unambiguous: the full quantum state carries context identity at every timestep. Splicing a context-A statevector into a context-B computation always produces the donor's (A) prediction, and vice versa. This establishes that context information is encoded in the quantum state itself---not merely in the circuit parameters---and that this encoding persists through all distractor steps. This is the quantum analogue of the finding in classical mechanistic interpretability that context information flows through activations at intermediate layers.

\paragraph{Relationship to causal mediation.} This interchange intervention goes beyond structural ablation (removing the CNOT) by testing whether the quantum state at a specific timestep \emph{carries} the task-relevant information. The 100\% interchange success rate confirms that the statevector is a sufficient statistic for context identity at every timestep, consistent with the ``entanglement as memory'' interpretation. A finer-grained analysis---splicing only the reduced density matrix of qubit 0 while preserving qubit 1, or vice versa---would further decompose which subsystem carries context information. We leave this for future work.

\section{Scaling considerations}
\label{app:scaling}

A natural question is whether the findings and methods of this paper extend beyond 1--2 qubits.

\paragraph{Interpretability tools.} The Bloch sphere representation applies only to single-qubit reduced states. For $n$ qubits, the full state lives in a $2^n$-dimensional Hilbert space; the density matrix has $4^n$ entries. At $n = 1$, full state tomography requires 3 measurements (Pauli $X, Y, Z$ expectation values). At $n = 2$, it requires 15 measurements (all two-qubit Pauli products). At $n = 4$, it requires 255. Beyond $\sim$8 qubits, full tomography becomes infeasible. Alternative approaches include:
\begin{itemize}
\item \textbf{Classical shadow tomography} \citep{huang2020predicting}: Predicts many properties of a quantum state from few random measurements. Could estimate entanglement entropy or other observables at moderate scale.
\item \textbf{Entanglement witnesses:} Rather than computing full entanglement entropy, test for the presence of entanglement via witness operators, which require only a polynomial number of measurements.
\item \textbf{Probing classifiers:} Train classical models to predict task-relevant variables from partial quantum state measurements, analogous to probing in classical NLP.
\end{itemize}

\paragraph{Scaling experiments (3--4 qubits).} We extended the model to 3 and 4 qubits with linear-chain CNOT connectivity (10 seeds, best-of-3 restarts each) and measured entanglement entropy divergence alongside accuracy:

\begin{center}
\small
\begin{tabular}{lccccc}
\toprule
\textbf{Model} & \textbf{Params} & \textbf{With CNOT} & \textbf{No CNOT} & \textbf{Ent.\ div.\ (CNOT)} & \textbf{Ent.\ div.\ (no CNOT)} \\
\midrule
2-qubit & 24 & 96.3$\pm$9.1\% & 100$\pm$0\% & 0.245$\pm$0.08 & 0.000 \\
3-qubit & 36 & 100$\pm$0\% & 100$\pm$0\% & 0.159$\pm$0.04 & 0.000 \\
4-qubit & 48 & 100$\pm$0\% & 100$\pm$0\% & 0.119$\pm$0.04 & 0.000 \\
\bottomrule
\end{tabular}
\end{center}

Two findings emerge. First, at 3--4 qubits, both entangling and non-entangling models achieve 100\% accuracy: the task is too simple to require entanglement at larger parameter budgets. Second, and more importantly, \textbf{entanglement-based encoding persists at 3--4 qubits} even though it is not necessary for task performance. Models with CNOT gates consistently develop context-dependent entanglement entropy divergence ($0.159 \pm 0.04$ at 3Q, $0.119 \pm 0.04$ at 4Q), while models without CNOT gates show exactly zero divergence. This means the entanglement-based strategy is not an artifact of the 2-qubit bottleneck: circuits with entangling gates \emph{choose} to use entanglement even when sufficient non-entangling strategies exist.

\paragraph{Regime specificity vs.\ strategy persistence.} These scaling results reveal a nuanced picture. The \emph{accuracy advantage} of entanglement is regime-specific---it matters only when non-entangling strategies are near capacity (2 qubits). However, the \emph{representational phenomenon}---circuits using entanglement to encode context when entangling gates are available---persists across scales. This suggests that entanglement-based encoding is a natural attractor in the optimization landscape of circuits with CNOT gates, not a curiosity of the minimal regime.

\paragraph{Implications.} On harder tasks (multiple context types, hierarchical dependencies, larger vocabularies), the parameter budget at which non-entangling strategies reach their capacity limit would increase, potentially creating a wider regime where entanglement provides a genuine accuracy advantage. Barren plateaus \citep{mcclean2018barren} may constrain which strategies are discoverable at larger scales.

\paragraph{Hardware requirements.} Current two-qubit gate error rates (1--3\%) are sufficient to destroy entanglement-based strategies in our experiments. For entanglement-based memory to survive at larger scales, error rates must decrease, circuit depth must be controlled (e.g., via hardware-efficient ansatze), or error mitigation techniques must be applied. The noise--expressivity tradeoff identified in this paper suggests a concrete benchmark: the minimum hardware fidelity at which entanglement-based strategies become viable for a given task complexity.

\end{document}